%% file: 00-bigdata.tex
\documentclass[10pt, conference, compsocconf]{IEEEtran}
\ifCLASSINFOpdf
\else
\fi

\usepackage[table,xcdraw]{xcolor}
\usepackage{graphicx}
\usepackage{amsmath}
\usepackage{mathtools}
\usepackage{amsfonts}
\usepackage{balance}
\usepackage{flushend}

\newtheorem{lemma}{Lemma}

\newcommand{\ignore}[1]{}

\begin{document}
%
\title{Contaminant Removal for Android Malware Detection Systems}
\author{\IEEEauthorblockN{Lichao Sun\IEEEauthorrefmark{1},
Xiaokai Wei\IEEEauthorrefmark{2},
Jiawei Zhang\IEEEauthorrefmark{3},
Lifang He\IEEEauthorrefmark{4},
Philip S. Yu\IEEEauthorrefmark{1} and
Witawas Srisa-an\IEEEauthorrefmark{5}}
\IEEEauthorblockA{\IEEEauthorrefmark{1}University of Illinois at Chicago, Chicago, IL \IEEEauthorrefmark{2}Facebook, Menlo Park, CA}
\IEEEauthorblockA{\IEEEauthorrefmark{3}IFM Lab, Florida State University, FL \IEEEauthorrefmark{4}Cornell University, New York City, NY}
\IEEEauthorblockA{\IEEEauthorrefmark{5}University of Nebraska - Lincoln, Lincoln, NE}
\IEEEauthorblockA{
Email: \{lsun29, xwei2, psyu\}@uic.edu, \{jwzhanggy, lifanghescut\}@gmail.com, witty@cse.unl.edu}
}

\maketitle

\input{01-abstract}

%
%

%
%

%
%



\IEEEpeerreviewmaketitle

\input{02-introduction}

\input{03-background}

\input{04-method}
\input{05-evaluation}
\input{06-results}
\input{07-relatedwork}

\input{08-conclusion}

\balance
\bibliographystyle{IEEEtran}
\bibliography{IEEEfull,sigproc.bib}

\end{document}

%% file: 01-abstract.tex
\begin{abstract}

A recent report indicates that there is a new malicious
app introduced every 4 seconds. This rapid malware
distribution rate causes existing malware detection
systems to fall far behind, allowing malicious apps to
escape vetting efforts and be distributed by even
legitimate app stores.  When trusted downloading sites
distribute malware, several negative consequences
ensue. First, the popularity of these sites would allow
such malicious apps to quickly and widely infect
devices. Second, analysts and researchers who rely on
machine learning based detection techniques may also
download these apps and mistakenly label them as benign
since they have not been
disclosed as malware. These apps are then used as part of their
benign dataset during model training and testing. The presence of 
contaminants in benign dataset can compromise
the effectiveness and accuracy of their detection and classification
techniques. 

To address this issue, we introduce {\sc PUDroid}
(Positive and Unlabeled learning-based malware
detection for Android) to automatically and effectively
remove contaminants from training datasets, allowing
machine learning based malware classifiers and
detectors to be more effective and accurate. 
To further improve the performance of such detectors, 
we apply a feature selection strategy to select pertinent features from a variety of features. We then compare the detection
rates and accuracy of detection systems using two
datasets; one using \textsc{PUDroid} to remove
contaminants and the other without removing
contaminants. 
The results indicate that once we remove contaminants from the datasets, we can significantly improve both malware detection rate and detection accuracy.
\end{abstract}


%% file: 02-introduction.tex
\section{Introduction} \label{introduction}

Android is currently the most used smart-mobile device
platform in the world, occupying 87.6\% of market
share and over 1.4 billion Android
devices in deployment~\cite{idc2016}.
Unfortunately, the popularity of Android also makes it
a popular target for cyber-criminals to create
malicious apps that can steal sensitive information and
compromise systems~\cite{forbes2014}. During the first
three months of 2016, Kaspersky Lab uncovered over 2
million malware samples including trojans, worms,
exploits, and viruses.  On average, a malicious app is
introduced in every 3.79 seconds~\cite{kaspersky2016}.
Some types of malicious apps have more than 50
variants, making detecting all of them very
challenging~\cite{symantec2016}.  

There have been several approaches to detect these
malicious Android apps. Most approaches focus on the
attack behaviors, and use static or dynamic analysis to
\begin{figure}[tbh!]
\centering
\includegraphics[width=0.45\textwidth]{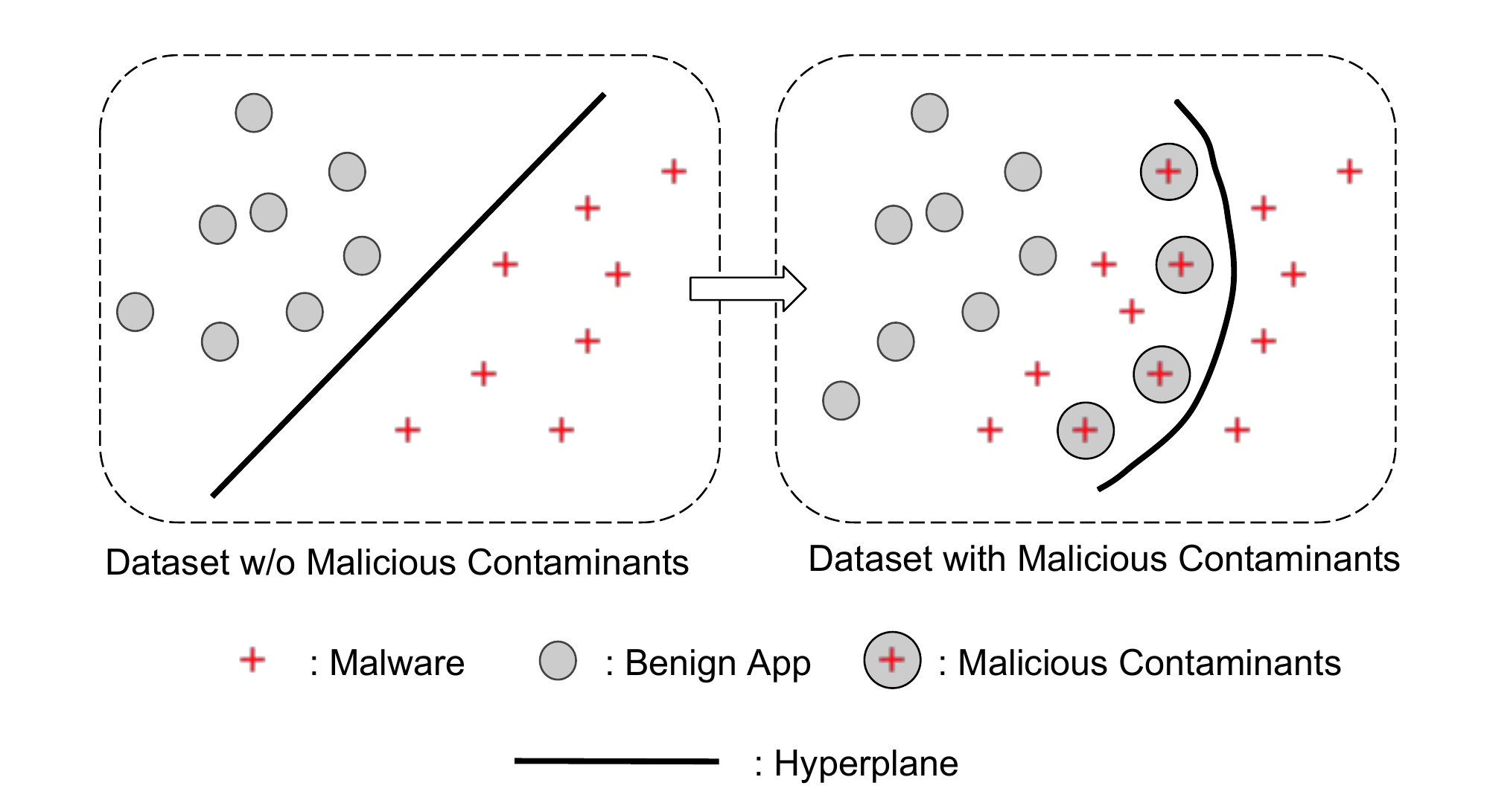}
\vspace{-15pt}
\caption{Left figure shows machine learning can
classify malware and benign apps well without malicious
contaminants. Right figure shows that the machine
learning cannot work well for malware detection with
malicious contaminants}
\label{fig:zero-day}
\vspace{-15pt}
\end{figure}
build detection tools that rely on approaches known to
work well for desktop
environments~\cite{grace2012riskranker}. 
However, static analysis approaches in general can
produce a large number of false positives while
dynamic analysis approaches need adequate input suites
to sufficiently exercise execution paths.  
Therefore, neither of them will work well 
for Andriod malicious app detection. 
Another emerging approach is to build 
detection techniques based on data mining and machine
learning techniques~\cite{arp2014drebin,sun2016sigpid,li2016droidclassifier}. 

For example, {\sc Drebin}~\cite{arp2014drebin} utilizes
multi-view features by combining static analysis and
supervised learning to accurately detect malware. {\sc
SIGPID}~\cite{sun2016sigpid} improves upon {\sc
Drebin}~\cite{arp2014drebin} by using many more
features for training and detection.  {\sc
DroidClassifier}~\cite{li2016droidclassifier} uses
traffic flow information and unsupervised learning to
detect the malware and classify the family of each
malicious app. 
 
\begin{figure*}[tbh!]
\centering
\includegraphics[width=0.85\textwidth]{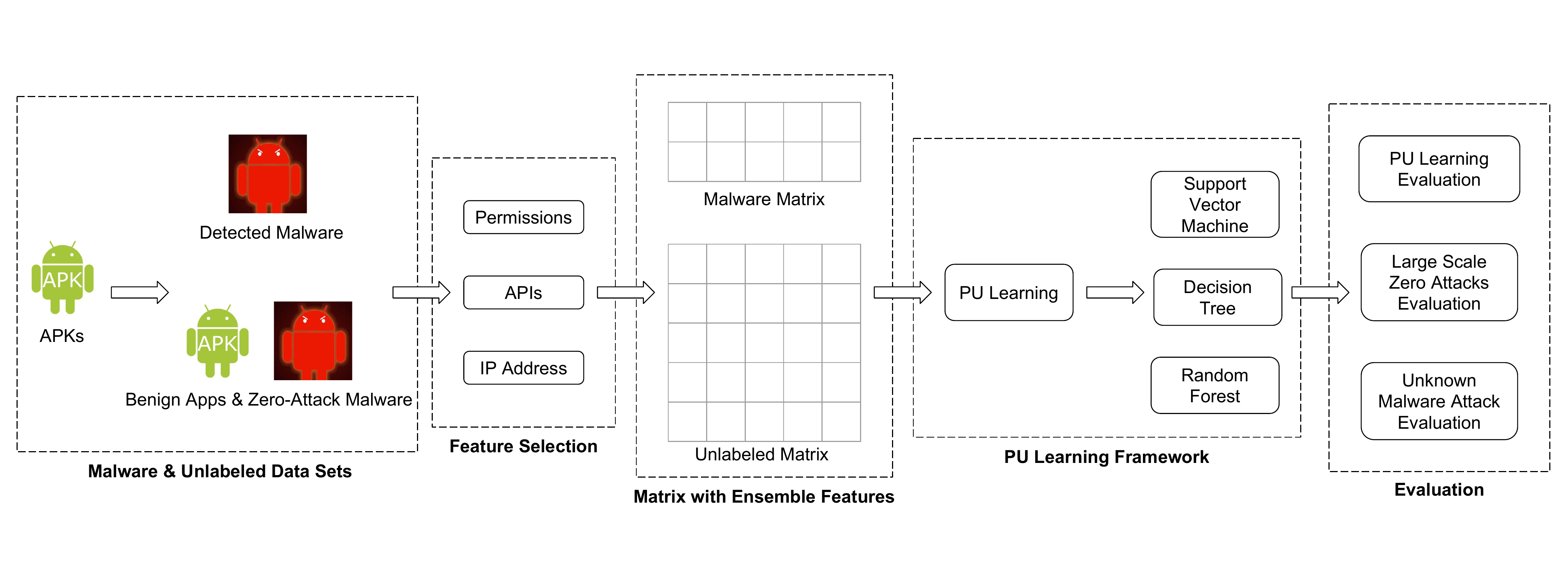}
\vspace{-18pt}
\caption{An Overview of PUDroid Approach}
\label{fig:PUDroid}
\vspace{-10pt}
\end{figure*}

When machine learning techniques are used to help with
malware detection, the detection effectiveness and
accuracy are highly dependent on the quality of the
training datasets. To create such dataset,
researchers typically label a set of malicious apps
and a set of benign apps. To build the malicious dataset, researchers manually label these malicious apps
one by one based on known information from various
malware analysis and collection sources (e.g., virusshare.com). To build the benign dataset,
researchers download apps from trusted
distribution sources such as Play Store and verify that those apps have not been recently disclosed as malware. However, as
previously mentioned, malicious apps are created every few seconds so malware detection sites often fall behind in disclosing new malware. As such, trusted distribution sites have
been known to distributed repackaged malware~\cite{Acharya16}.

Currently, repackaging benign apps with existing malicious components is the leading approach employed by cybercriminals to create
Android malware~\cite{idc2016}. Thus, it is quite common for undisclosed malware to have previously known basic attack behaviors. However, the sheer volume of recently created malware also makes it possible for benign dataset to contain
undisclosed malware as illustrated in the right side of
Fig.~\ref{fig:zero-day}. The presence of these malicious apps in a benign dataset (we refer to these
malicious apps as contaminants) can negatively affect
the accuracy and effectiveness of machine
learning-based detection tools. As such, we need a
mechanism to effectively remove such contaminants from
benign dataset. 


In this paper, we present {\sc PUDroid}, an approach that leverages the \emph{Positive and Unlabeled} (\emph{PU}) learning to remove
contaminants as part of building benign dataset for robust malware detection in Android. Specifically, we assume that the positive group contains malicious apps, and the unlabeled group contains benign apps and some contaminants consisting of unlabeled malicious apps. We make use of a feature selection strategy to build the desired input for {\sc PU} learning. In a nutshell, {\sc PUDroid} works as an additional validator to ensure that presumably benign apps downloaded from trusted distribution channels are indeed benign. 
That is, they do not contain any undisclosed but detectable malicious components due to repackaging. We then conduct empirical evaluation of \textsc{PUDroid} and find that it allows to remove nearly 100\% of all contaminants as part of building benign dataset. 

An overview of our {\sc PUDroid} approach is illustrated in Fig.~\ref{fig:PUDroid}. The main contributions are summarized as follows:

\vspace{5pt} \noindent
(1) To the best of our knowledge, {\sc PUDroid} is the first
attempt to use PU learning to remove contaminants as
part of building benign dataset for robust malware detection in Android.
    
\vspace{5pt} \noindent (2) We introduce a feature
selection strategy to analyze and select
``explainable'' features that are useful for effective
malware detection. Our strategy reduces the number of
embedding features by 93\% when we compare to those
used by existing approaches, while yielding nearly the
same level of effectiveness in malware detection.

\vspace{5pt} \noindent
(3) We evaluate {\sc PUDroid} using a
large dataset containing 5,560 malware samples with
2,200 features. We also investigate a scenario in which
{\sc PUDroid} needs to detect different magnitudes and types of
contaminants. The results show that {\sc PUDroid} is very
effective in such situations. 
    
\vspace{5pt} \noindent
(4) We also evaluate {\sc PUDroid} with different
classification methods to fine tune the accuracy of
{\sc PUDroid}. The results indicate that the
performance of a classification model is sensitive to
the number of malware samples present in a dataset.


%% file: 03-background.tex
\section{Background And Motivation} \label{background}

In this section we briefly introduce common machine
learning approaches employed by researchers to build
malware detection frameworks for Android. These
techniques employ many features ranging from static
information such as permissions to dynamic information
such as network traffics. 

\subsection{Permission-Based Learning}

Android security system offers permission control
mechanisms as one of most important components. Android
app developers need to declare the permissions for each
app to have access to resources such as text messaging
and address book. The declaration can be found in the
manifest file. When users try to install apps, they can
choose to approve or decline the requested permissions.

In the case of installing apps from third party stores,
users may need to root their devices, making their
systems more vulnerable. Furthermore, we have also seen
that developers tend to request more permissions than
the apps actually need. These behaviors make Android
security mechanism vulnerable to malicious
attacks. To help a user
determine if requested permissions can make a system
vulnerable to attacks, Android provides "protection
Level" to help characterize the potential risks of
these permissions~\cite{android2016manifest}.  The four
protection levels are ``normal'', ``dangerous'',
``signature'', and ``signature or system''.
Permissions such as "WAKE\_LOCK" which keep processor
from sleeping or screen from dimming are considered as
low risk. However, some permissions such as
"WRITE\_SMS" are considered dangerous due to its
ability to leak information and cause financial damages
such as texting to premium services. Work by Sun et
al.~\cite{sun2016sigpid} finds that normal permissions,
in addition to those classified as ``dangerous'' can
also significantly contribute to malicious behaviors.

\subsection{API-Based Learning}

Previously, we introduced the Android permission system
and how it can be used to detect the potential malware.
However, permission-based approaches may not always
yield accurate results as many permissions are commonly
used by both benign and malicious apps. To improve
accuracies of these approaches, API (Application
Programming Interface) information can also be used to
add context that can help distinguishing between benign
and malicious apps~\cite{arp2014drebin}. For
example, my analyzing calling contexts leading to
dangerous APIs such as \texttt{WRITE\_SMS}, one can
compare benign calling contexts and malicious calling
contexts. 

\subsection{Network Information Based Learning}

Dynamic information such as network traffic can be used
as learning feature to detect the malware. To do so,
input generation techniques such as Monkey are used to
generate inputs to execute benign and malicious
programs. The network traffic is then recorded and then
analyzed as features for machine learning. For example,
Shabtai et al.~\cite{shabtai2012andromaly} present a
Host-based Android malware detection
system to target the repackaging attacks. They conclude
that deviations of some benign behaviors can be
regarded as malicious ones.  Narudin et
al.~\cite{narudin2016evaluation} introduce a
TCP/HTTP based malware detection system. They extracted
basic information, (e.g. IP address), content based,
time based and connection based features to build the
detection system. Their approach can determine if an
app is malicious or not.

\subsection{Motivation}
In this paper, we propose a {\sc PU} learning-based approach (named as {\sc PUDroid}) that aims to remove contaminants as part of building benign dataset for robust malware detection in Android. We show that {\sc PUDroid} can protect the system when the benign dataset are infected by malware. To improve the performance of the {\sc PUDroid} and to prevent various malicious behavior, we also use embedding features including permissions, APIs, and network information to efficiently detect the malware.

%% file: 04-method.tex
\section{PUDroid Approach}
In this section, we introduce the main steps of our {\sc PUDroid} approach including feature selection and {\sc PU} learning. First, we present a feature selection strategy to help us find the informative features from a variety of Android features. Then, we use the resulted features to leverage {\sc PU} learning process along with different classifiers to remove contaminants for robust malware detection.

\subsection{Feature Selection for Android Data Generation} \label{feature}
Android apps can collect malicious and benign datasets, which usually contain various types of features such as permissions, APIs, IP address, activities, requested URLs, and services. However, not all of these features are effective for malware detection and many features such as activities, requested URLs, and services are difficult to explain due to large variability in these features. We need to perform feature selection to find the informative features from these available features. By using domain acknowledge of security, 
our feature selection strategy only uses explainable and helpful features to detect malware. In particular, we only focus on permissions, APIs, IP address and URLs. We first convert the URLs into IP address, and then introduce our feature selection strategy on how to select features from permissions, APIs and IP address features.

IP address can provide similar information as URLs but with fewer variability. For example, multiple URLs can point to the same IP address if they are alias. Part of IP address, e.g., the three most significant bytes, can
also provide company information (e.g., 216.59.192.xx tells us that the IP address in this range belong to
Google). Additionally, note that both IP address and URLs can suffer from spoofing but IP address provides the same information with much fewer data points. Therefore, it is beneficial to convert the URLs to a more effective representation (IP address). To achieve this goal, we remove the invalid URLs as they tend not to have corresponding IP address, and use the socket API to acquire IP address.

Based on the above results, we define an unbalanced feature selection strategy for malicious and benign datasets. A large difference between the numbers of malware and benign samples can lead to unbalanced removal of less contributing features. To overcome this issue, our feature removal process biases the criteria based on the ratio between the number of benign samples and malware samples. 
For example, if there are twice as many benign samples as malware samples, we then remove a seldom occurring feature from
malware dataset based on a threshold $tm$. For example, if feature occurs 5 times when we consider the entire
malware dataset, we would then set $tb$ for removal for the benign dataset to two (i.e., twice as many as $tm$
or 10 times). This can be formulated as:
\begin{align}
    \frac{tm}{tb} = \eta \cdot \frac{\# benign \; samples}{\# malware \; samples}
\end{align}
where $\eta$ is the coefficient that controls threshold selection and 
it is usually at least 2. Here we set $\eta=2$, which means any feature must be used at least two malware and $tm/tb$ benign apps will not be removed after feature selection.

After applying above feature selection strategy, we use $\mathbf{x}_1$, $\mathbf{x}_2$, and $\mathbf{x}_3$ to represent the selected feature subsets of permissions, APIs, and IP address, respectively. Our final resulted feature set is:
\begin{align*}
    \mathbf{x}:= \mathbf{x}_1 \cup \mathbf{x}_2 \cup \mathbf{x}_3
\end{align*}




\subsection{PU Learning For Malware Detection} \label{method}

As mentioned earlier, unreliable negative examples (or
contaminants) can be unknowingly included in training
dataset due to the prolific rate of malware
creation~\cite{Acharya16}. 
We leverage {\sc PU} learning to detect and remove these contaminants from training datasets. 
Table~\ref{table:symbols} lists the basic symbols that will be used throughout this section.

\begin{table}[]
\centering
\caption{List of Basic Symbols}
\label{table:symbols}
\begin{tabular}{|c|l|}
\hline
Symbol  & Definition and description \\ \hline
$\mathbf{P}$    & positive group/marked malware set                                                                                                             \\ \hline
$\mathbf{U}$    & unlabeled group/mixed malware and benign apps                                                                                                  \\ \hline
$\mathbf{1}$    & a vector of all ones                                                                                                                                \\ \hline
$\mathbf{x}(s)$    & resulted feature vector of an app s                                                                                                                                \\ \hline
$z(s)$    & 1 means the app s is labeled, 0 otherwise                                                                                                  \\ \hline
$y(s)$    & 1 means the app s is true malware, 0 otherwise                                                                                                  \\ \hline
$p(\cdot)$ &the probability of an app \\ \hline
$f(\cdot)$ & malicious probability of an app without {\sc PU} learning   \\ \hline
$g(\cdot)$ & malicious probability of an app with {\sc PU} learning    \\ \hline
$M_d$   & \begin{tabular}[c]{@{}l@{}}classifier model without {\sc PU} learning, if $f(\cdot)$\textgreater0.5,\\ the app is malicious, otherwise is benign\end{tabular} \\ \hline
$M_h$   & \begin{tabular}[c]{@{}l@{}}classifier model with {\sc PU} learning,if $g(\cdot)$\textgreater0.5, \\ the app is malicious, otherwise is benign \end{tabular}    \\ \hline
$\sim$ & denotes the equivalent relation                                                                                           \\
\hline
\end{tabular}
\vspace{-5pt}
\end{table}


{\sc PU} learning is a semi-supervised technique for building a binary classifier on the basis of positive and unlabeled samples only. It is useful when we have not determined if an app is malicious or benign. In small collections of apps, labeling an app as malicious or benign can be done manually without too much effort. In large collections, however, manual identification may not be feasible. 
In this case, {\sc PU} learning can be applicable. 
In addition, in a scenario that the identification process
may produce inaccurate results (e.g., mistakenly
identifying malicious apps as benign apps), {\sc PU}
learning can also help to identify and remove these
unreliable negative samples.

In order to use {\sc PU} learning for malware detection, we divide our dataset into the positive group ($\mathbf{P}$) and the unlabeled group ($\mathbf{U}$), where the positive group contains malicious apps, and the unlabeled group contains benign apps and some contaminants consisting of unlabeled malicious apps.
To differentiate positive and unlabeled apps, we define ``discovery state" ($z$) to indicate whether an app is labeled or not in the
dataset. For a given app $s$ in the group $\mathbf{P}$, if $s$ is marked as malicious, then $z(s) = 1$; otherwise, $z(s) = 0$. As a result, 
the ``discovery states" of apps in groups $\mathbf{P}$ and $\mathbf{U}$ are: $z(\mathbf{P}) = \mathbf{1}$ and $z(\mathbf{U}) = \mathbf{0}$.  
Besides, each app has another label called ``hidden malware state" ($y$), which can expose whether an app is actually malicious or benign. Here an app with 1 is a known malware app, and 0 is a benign app. For example, if an app $s$ has been detected as malware,
then $y(s) = 1$; otherwise $y(s) = 0$. Based on this, we can check the ``hidden malware state" of each app in
group $\mathbf{P}$ and group $\mathbf{U}$. Since every app in group $\mathbf{P}$ is a known malware app, $y(\mathbf{P}) = \mathbf{1}$.
However, $y(\mathbf{U})$ can be either $\mathbf{1}$ or
$\mathbf{0}$ as both malicious contaminants and benign
apps can be in $\mathbf{U}$. As such, we have:
\begin{equation}\label{func1}
p(z(s) = 1|\mathbf{x}(s), y(s) = 0) = 0.
\end{equation}
\noindent
where $p(\cdot)$ is the probability, and $\mathbf{x}(s)$ is the feature vector extracted
for app $s$ (i.e., [0, 1, 0, 0,..., 1], 1 denotes that the app requests the
permission and 0 otherwise).

The goal of malware detection is to build a malware discovery model $M_d
\sim f(\mathbf{x}(s)): \mathbb{R}^d \rightarrow \{1, 0\}$ from $\mathbf{P}$ and $\mathbf{U}$, where $d$ is the dimension of the
$x(s)$, which represents the number of the features. $1$ means $s$ is identified as malware, $0$ means $s$
is not identified as malware. Given an app $s$, by applying $M_d \sim f(\mathbf{x}(s))$, the discovery
probability of $s$ as a malware is: $p(z(s) = 1 |
\mathbf{x}(s))$. While our ultimate goal is to infer the true label of a given malicious app $s$ (i.e., $y(s)$). Besides the \textit{discovery probability}, we also need to build a hidden malware detection model $M_h \sim g(\mathbf{x}(s))$ based on the detected malware and the unlabeled app sets. Formally, we define the
probability that an app $s$ is indeed malicious (i.e., $y(s) = 1$) as:
$p(y(s) = 1|\mathbf{x}(s))$. In the following, we discuss the details.

\textbf{Assumption}: (Malware Discovered at Random): Assume malware samples are randomly detected by analysis, the
probability of detection is not relevant with the features created, then it has: 
\begin{equation}\label{func2}
   p(z(s) = 1|\mathbf{x}(s); y(s) = 1 ) = p(z(s) = 1 | y(s) = 1).
\end{equation}

To build the proposed hidden malware detection model
$M_h$, we should ideally know which apps are actually
malware. For this purpose, we first prove Lemma 1~\cite{elkan2008learning}.
\begin{lemma}
Suppose the ``Malware Discovered at Random'' assumption holds, then
\begin{equation}\label{func3}
p(y(s) = 1|\mathbf{x}(s)) =\frac{p(z(s) = 1|\mathbf{x}(s))}{p(z(s) = 1|y(s) = 1)}.
\end{equation}
\end{lemma}
\begin{proof}
By holding the assumption, we have:
\begin{align} \label{func4}
 ~~~~~~~& p(z(s) = 1|\mathbf{x}(s)) \nonumber \\
~~~~~~~ = & ~p(z(s) = 1|\mathbf{x}(s)) \cdot p(y(s) = 1|\mathbf{x}(s), z(s) = 1) \nonumber\\
~~~~~~~    = & ~p(y(s) = 1, z(s) = 1|\mathbf{x}(s)) \nonumber\\
~~~~~~~    = & ~p(y(s) = 1|\mathbf{x}(s)) \cdot p(z(s) = 1|y(s) = 1, \mathbf{x}(s)) \nonumber\\
~~~~~~~    = & ~p(y(s) = 1|\mathbf{x}(s)) \cdot p(z(s) = 1|y(s) = 1).
\end{align}
Then, by dividing both sides of Eq.~(5) by $p(z(s) = 1|y(s) = 1)$, we arrive at Lemma 1.
\end{proof}

From Eq.~(4), if we want to build $M_h$ by
calculating the \textit{hidden malware probability}, we
can use $f(\mathbf{x}(s))/p(z(s) = 1|y(s) = 1)$, where
$f(\mathbf{x}(s))$ is the malware probability of the an
app $s$ of $M_d$. Specifically, we have: 
\begin{equation}\label{func5}
   f(\mathbf{x}(s)) = p(z(s) = 1|\mathbf{x}(s)).
\end{equation}

On the other hand, $p(z(s) = 1|y(s) = 1)$ can be calculated by
a validation set with applying $M_d$. First, we
randomly select apps from $\mathbf{P} \cup \mathbf{U}$ to be the
validation set $\mathbf{V}$, and choose the subset $\mathbf{P}'$ from $\mathbf{V}$
with the positive label, i.e. $z(\mathbf{P}') = 1$. The
estimator of $p(z(s) = 1|y(s) = 1)$ is the average
value of $g(\mathbf{x}(s))$ for $\mathbf{x}(s)$ in
$\mathbf{P}'$, then we have:
\begin{equation}\label{func6}
p(z(s) = 1|y(s) = 1) \sim   e = 1/n \cdot \sum_{x \in \mathbf{P}'}f(\mathbf{x}(s)).
\end{equation}
where $n$ is the cardinality of $\mathbf{P}'$ and 
the estimator $e$ is the average value of $f(\mathbf{x}(s))$ for $x$
in $\mathbf{P}'$. Since $e$ is based on a certain number of data
instances, it has a low variance and is preferable in practice \cite{elkan2008learning}.
Notice that to compute $e$, we need to specify the size of the set $\mathbf{P}'$ and the size of
the validation set $\mathbf{V}$. We explain how to set the size
and evaluate our system in Section~\ref{evaluation}.

With $e$ and the classifier model $M_d \sim f(\mathbf{x}(s))$ on labels $z(s)$, 
we can adjust to a classifier $M_h \sim g(\mathbf{x}(s))$  
on relation labels $y(s)$ as follows:
\begin{align}\label{func7}
M_h \sim g(\mathbf{x}(s)) &= p(y(s) = 1|\mathbf{x}(s)) \nonumber \\
   &= \frac{f(\mathbf{x}(s))}{e}.
\end{align}

\begin{figure}[t]
    \centering
    \includegraphics[width=0.50\textwidth]{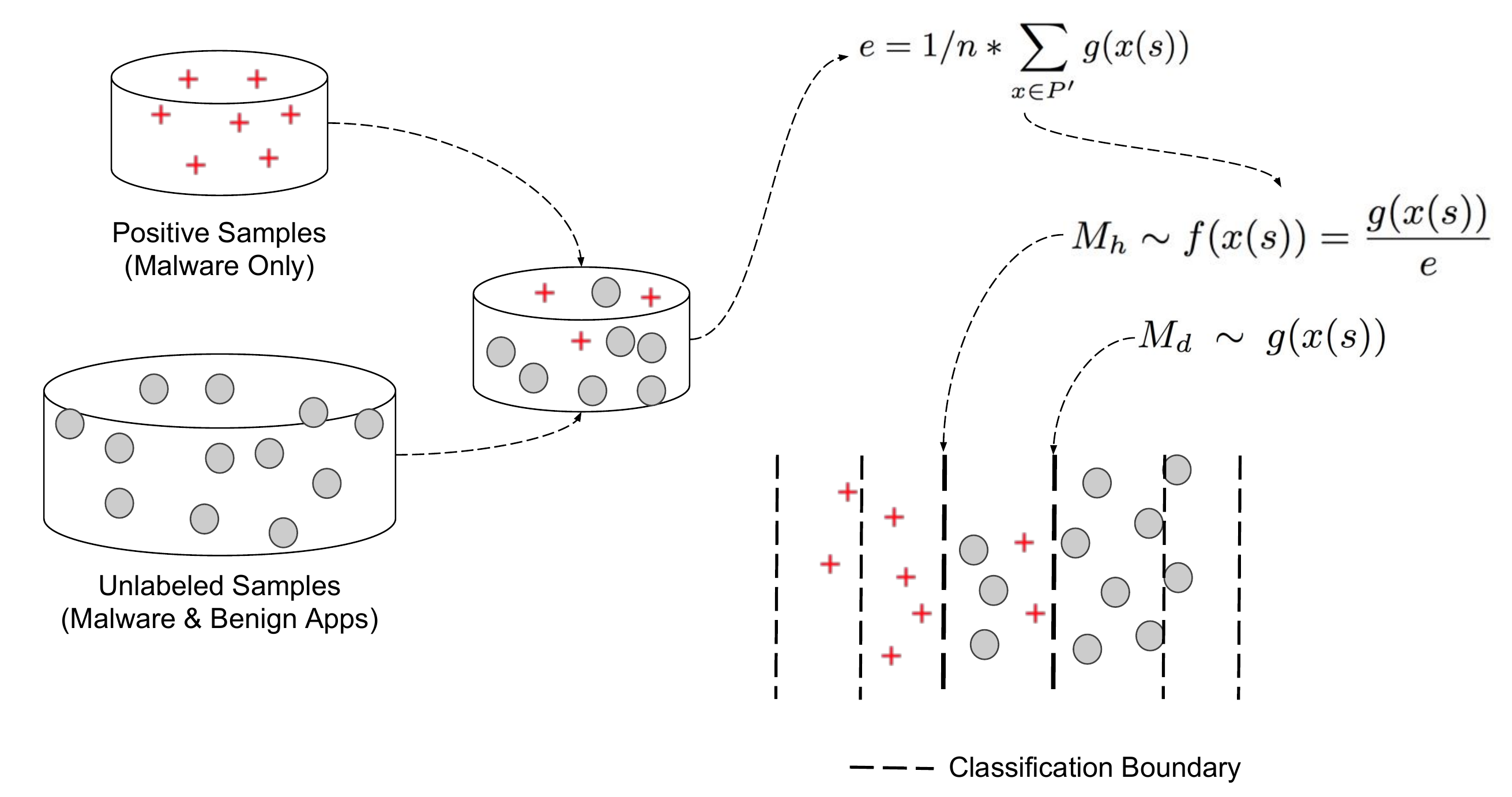}
    \caption{An Overview of Our PU Learning Process}
    \label{fig:PULearning}
     \vspace{-5pt}
\end{figure}


    

Fig.~\ref{fig:PULearning} shows an overview of our {\sc PU} learning process.

\textbf{Discussion}: Here we give an example to show
how the Eq.~(\ref{func7}) works for positive and
unlabeled datasets. If all apps can be clearly
separated into malicious and benign groups, (i.e., no
malicious apps in the unlabeled group) the malicious
probability ($g(\mathbf{x}(m))$) of most random malware
$m$ in positive group should be close to 1. However,
when the unlabeled group contains a large number of
malicious contaminants, the malicious probability of
most random malware $m$ in positive group will
decrease. The worst case is that the $g(\mathbf{x}(m))$
close or even less than 0.5. We know that if
$g(\mathbf{x}(m)) < 0.5$, the malware $m$ will be
classified as a benign app by the detector. 

To address this problem, when we first use a training set
to build the detector, we select a subset $\mathbf{P}_M$
from positive group where $y(\mathbf{P}_M) = \mathbf{1}$. Then
we calculate the average malicious probability of $P_M$
by using $g(\mathbf{x}(\mathbf{P}_M))$. If we find the average
malicious probability of $\mathbf{P}_M$ is close to 0.5 rather
than 1.0, we increase the malicious probability
for each app in the training set. Then more malicious
contaminants will be detected from the unlabeled group.
If the system is faced with malicious contaminants,
the malicious probability of every app will decrease
including benign apps by the detector. So most benign
apps will still be classified as benign after we
increase the malicious probability for them. This is
the framework of {\sc PU} learning to prevent the malicious
contaminants.



In order to apply {\sc PU} learning for the dataset with
group $\mathbf{P}$ and group $\mathbf{U}$, we need to implement the
learning algorithms to build the malware detection
system. Various learning algorithms have been used in
mobile security
before~\cite{sun2016sigpid},
we choose three most frequently used learning
algorithms: Support Vector Machine, Decision Tree and
Random Forest.

%% file: 05-evaluation.tex
\section{Evaluation} \label{evaluation}
We conducted experiments to
evaluate the performance of {\sc PUDroid} in
identifying contaminants in Android datasets.
We conducted our experiments using a large-scale
real-world collection of apps.  
Our experiments were
done to answer four research questions (RQs). 

\vspace{3pt}\noindent \textbf{RQ 1:} How effective is
\textsc{PUDroid} in removing malicous contaminants from
benign dataset? In this experiment, we create
contaminated benign dataset and then compare the
performances between a system with \textsc{PUDroid} and
without \textsc{PUDroid}. 

\vspace{3pt}\noindent \textbf{RQ 2:} How do changes in
the number of malicious contaminants in benign data
sets affect the performance of \textsc{PUDroid}? 
In this experiment, we systematically increased the
number of malicious contaminants in our benign datasets. 

\vspace{3pt}\noindent
\textbf{RQ 3:} How effective is \textsc{PUDroid} in
removing unknown malicious contaminants from benign
dataset? In this experiment, we include a family of
malware into each benign dataset and evaluate if
\textsc{PUDroid} can detect these contaminants. 

\vspace{3pt}\noindent \textbf{RQ 4:} How effective is
\textsc{PUDroid} in removing benign contaminants from
malicious dataset?  Our first two focuses were on
building accurate benign dataset. In this experiment,
we reverse the contamination pattern by including
benign apps in the malicious dataset and observe the
performance of \textsc{PUDroid} to detect these benign
contaminants.

\vspace{3pt}
\noindent
In all experiments, we then compared our results with
those from several state-of-the-art detection methods.
\subsection{Data Sets}
We used the dataset based on prior work by Arp et
al.~\cite{arp2014drebin}. The dataset includes 5,560
malware and 123,453 benign apps. We then extracted
relevant features using a strategy that differs from
theirs. In their work, they used many features
including permission and API information. However, not
all features that they used are effective for malware
detection and many features such as activities,
requested URLs, and services are difficult to explain
due to large variability in these features (e.g.,
there are over 10,000 activities and services and
200,000 URLs in the dataset). Having too many features
can result in very sparse matrix, leading to high
machine learning overhead and over-fitting issues.
Then we can apply feature selection to select informative 
features.




After feature selection, our system kept 2,200
features.  Drebin, on the other hand, used over 300,000
features.  Even with the reduced number of features,
our dataset contains too many dimensions. To address
this problem, we use
PCA to project the whole dataset to 2-dimension representation of the dataset,
Fig.~\ref{fig:pca} shows the distribution of our whole
dataset. We can find most malware and benign apps can
be separated even with low dimension representation
with our features. That's the reason the machine
learning can help to sepearate the malware and benign
apps when dataset is without any contaminants.

\begin{figure}[tbh!]
    \centering
    \includegraphics[width=0.39\textwidth]{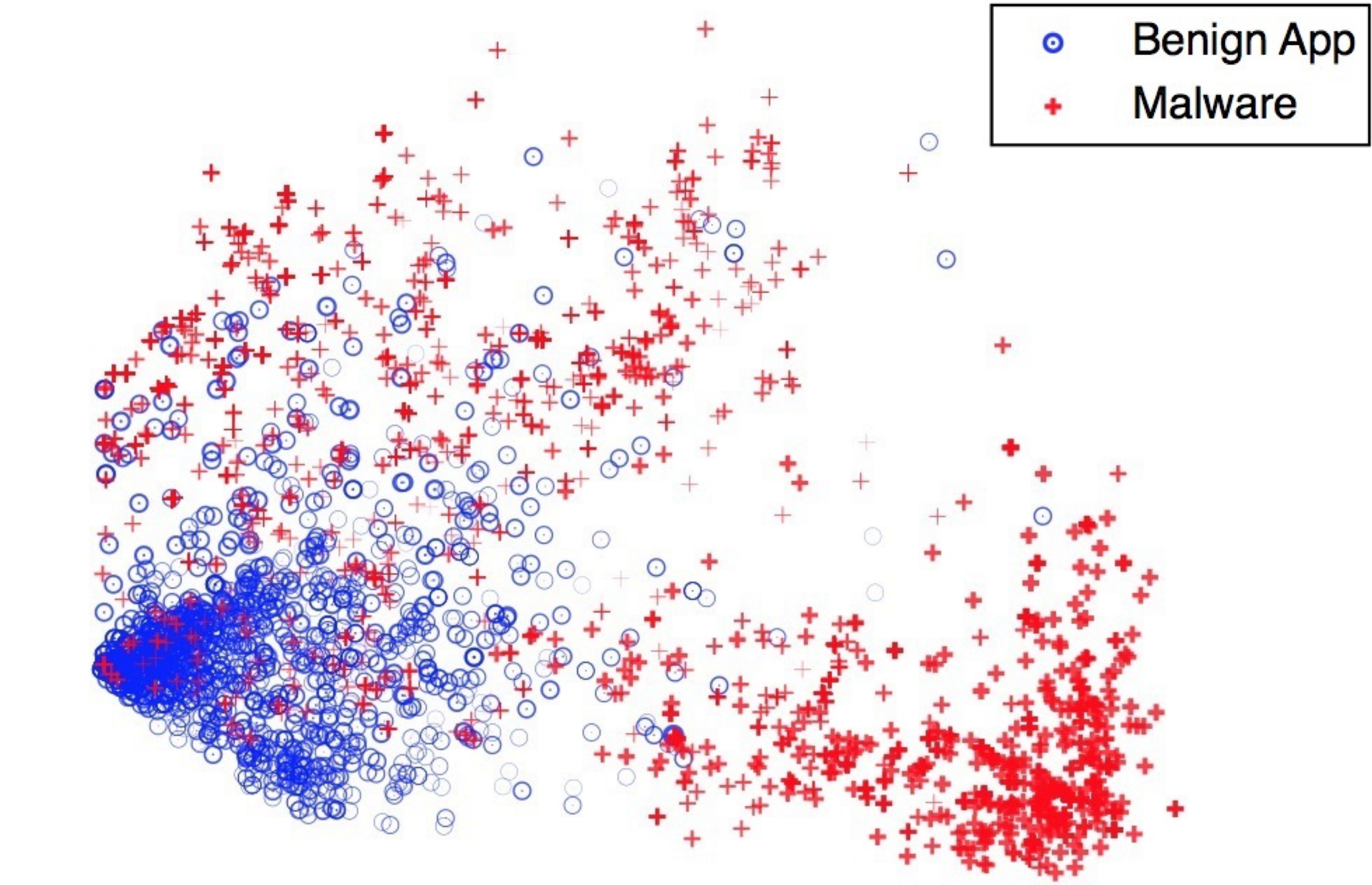}
    \caption{PCA Distribution of Whole Dataset}
    \label{fig:pca}
    \vspace{-5pt}
\end{figure}

\begin{figure*}[tbh!]
\begin{center}
\mbox{
\begin{minipage}[t]{.33\textwidth}
\includegraphics[width=1.0\textwidth]{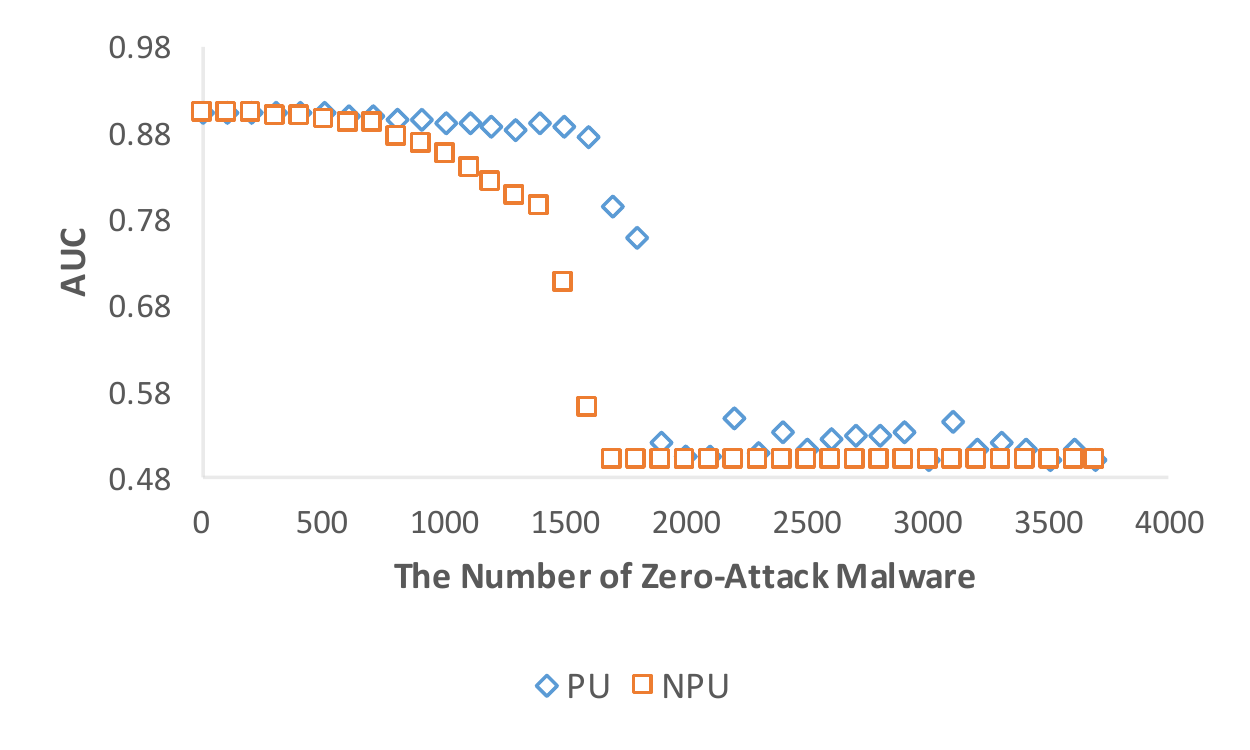}
\end{minipage}
\begin{minipage}[t]{.33\textwidth}
\includegraphics[width=1.0\textwidth]{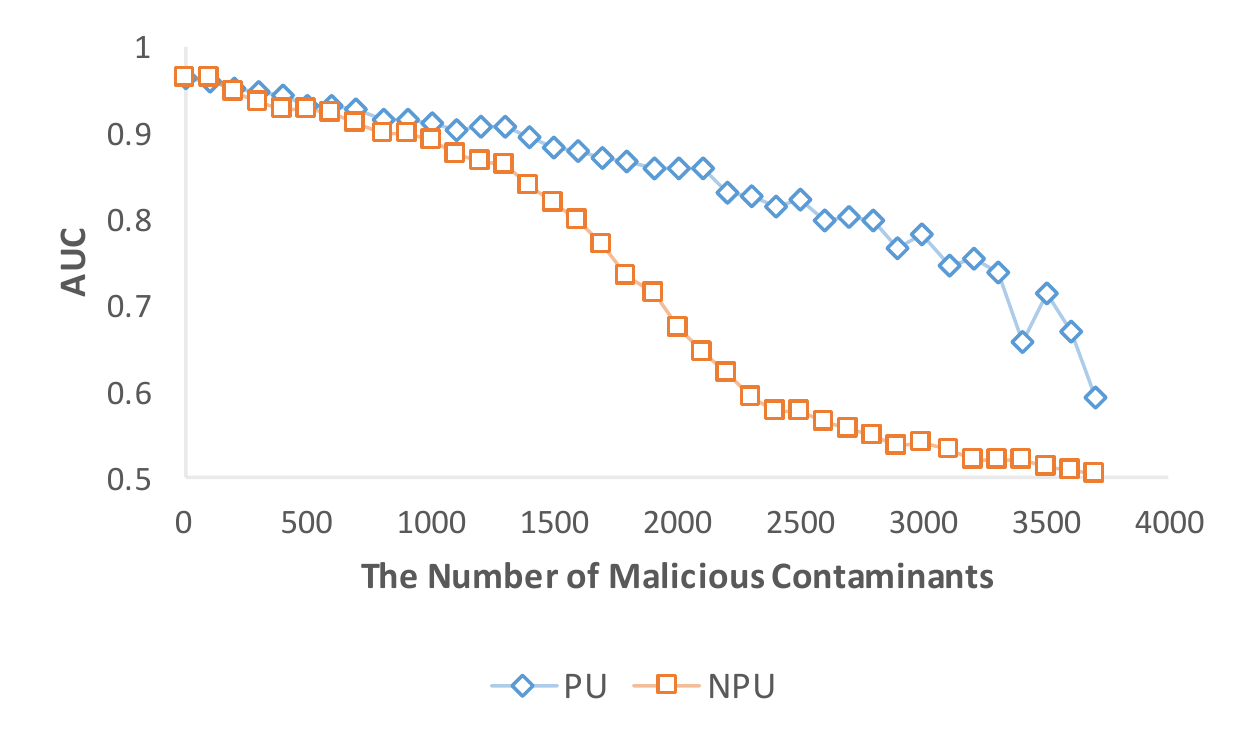}
\end{minipage}
\begin{minipage}[t]{.33\textwidth}
\includegraphics[width=1.0\textwidth]{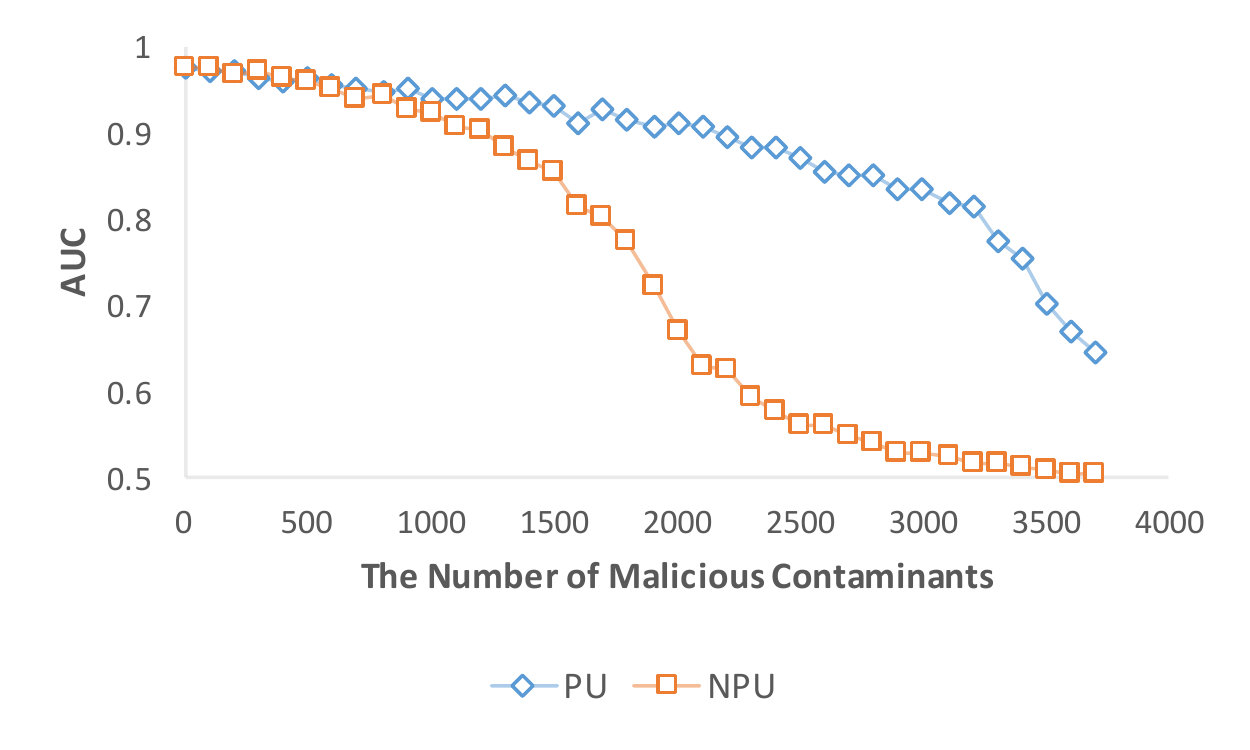}
\end{minipage}
}
\caption{AUC of PUDroid with SVM (left), Decision Tree
(center), Random Forest (right)}
\label{fig:results}
\end{center}
\vspace{2pt}
\end{figure*}


%% file: 06-results.tex
\begin{table*}[tbh!]
\centering
\caption{Evaluation of Highly Contaminated Datasets}
\label{table:largescale}
\begin{tabular}{|c|c|c|c|c|c|c|}
\hline
Scale Rate & Random Forest/PU & Random Forest/NPU & Difference & Decision Tree/PU & Decision Tree/NPU & Difference \\ \hline
100\%      & 83.32\%          & 44.25\%             & 39.07\%    & 73.18\%          & 47.06\%           & 26.12\%    \\ \hline
200\%      & 76.39\%          & 14.30\%             & 62.10\%    & 62.81\%          & 16.46\%           & 46.36\%    \\ \hline
300\%      & 71.51\%          & 8.69\%              & 62.82\%    & 57.15\%          & 11.83\%           & 45.32\%    \\ \hline
400\%      & 64.91\%          & 6.19\%              & 58.72\%    & 51.07\%          & 7.97\%            & 43.10\%    \\ \hline
500\%      & 64.74\%          & 3.32\%              & 61.42\%    & 49.70\%          & 5.86\%            & 43.85\%    \\ \hline
600\%      & 59.70\%          & 3.47\%              & 56.24\%    & 47.23\%          & 4.79\%            & 42.44\%    \\ \hline
700\%      & 54.22\%          & 2.50\%              & 51.71\%    & 44.32\%          & 4.66\%            & 39.65\%    \\ \hline
800\%      & 52.14\%          & 2.14\%              & 50.00\%    & 43.54\%          & 3.82\%            & 39.72\%    \\ \hline
\end{tabular}
\vspace{2pt}
\end{table*}

\section{Results}
\label{results}
In this section, we report the experimental results to
answer the four proposed research questions. We used three measures, \emph{accuracy}, \emph{AUC} and \emph{F-measure} to evaluate the performance of {\sc PUDroid}.

\subsection{RQ 1: Effectiveness of PUDroid} 
In order to evaluate {\sc PUDroid}, we first choose 1/3
of all malware samples and 1/3 of benign apps as the
testing set. We then use the remaining malware and
benign apps as the training set. To create an unlabeled
group, we also include malware into the benign apps in
multiple iterations using the following process.
Initially, we have no malicious contaminants, meaning
there is no unlabeled samples in this case. Every
sample is negative in the unlabeled group, which indicates
benign sample. In the next iteration, we randomly
select 100 malware samples from the training set as the
malicious contaminants. We then add these malware
samples to the benign portion of the training set and
remove them from the malware portion of the training
set. That is, we make these samples negative. Now, we
have a positive group and unlabeled group which
contains both malware and benign apps. We then repeat
above step and randomly select more malware samples in
the subsequent iterations. In summary, to build
unlabeled dataset, we randomly include in the benign
dataset, $100 \cdot N$ malware samples as malicious
contaminants in the $N_{th}$ iteration. 

We report the performance of \textsc{PUDroid}
using our dataset in Fig.~\ref{fig:results}. Note that, due to limited space, we only report \emph{AUC} values, but the conclusion is consistent across other measures. Our evaluation is based on the comparison of the performance of {\sc PUDroid} and a classification
method (denoted as PU in the figure) with that of using
the same classification method alone (denoted in the figure as No PU or NPU). We used all three classification learning methods. Based on the results, we can make the
following conclusions.

\begin{itemize}
\vspace{2pt}  
\item As shown in Fig.~\ref{fig:results}, {\sc PUDroid}
can work well with all three classification methods.
\vspace{2pt}      
\item If there are no malicious contaminants in the
testing dataset, {\sc PUDroid} yields the same
performance as using methods without PU learning.
\vspace{2pt}      
\item The performance of {\sc PUDroid} with Support
Vector Machine (SVM) is initially stable but degrades
quickly when the training dataset contains a large
number of malicious contaminants.
\vspace{2pt}      
\item The performance of Random forest is slightly
better than that of Decision Tree.
\vspace{2pt}      
\item When the dataset contains only a small number of
malicious contaminants, both of {\sc PUDroid} system
and methods without PU learning work well. 
    
\end{itemize}


\subsection{RQ 2: Variable Levels of
Contaminants}

In this section, we evaluate the malware detection
performance of {\sc PUDroid} when the benign dataset
contains a large portion of contaminants. To do so, we
apply the following method to create the unlabeled dataset. 

In $N_{th}$ case, we set the number of malicious
contaminants in the training dataset to be $N$ times
the number of malware in the training dataset. For
example, in the first case, we set the ratio of the number
of malicious contaminants and the number of malware in
the training set to be 1:1. In the $N_{th}$ case, the
ratio becomes $N:1$. The resulting unlabeled dataset
should test our system's ability to deal with a large
number of unlabeled malware samples.

In previous section, we showed that {\sc PUDroid} with
SVM can work well when the number of malicious
contaminants is small to moderate. However, it does not
work well when the number of malware is large. As such,
we only evaluate Random Forest and Decision Tree in
this scenario.

From Table~\ref{table:largescale}, we evaluate the ratios from 1:1 to 8:1, and report the malware detection rate (also referred to as True Positive Rate or Recall). Based on the results, we have the following observations:

\begin{itemize}
\vspace{2pt}  
\item Random Forest or Decision Tree alone cannot work
well when they face a large number of malicious
contaminants. For example, when the ratio is 3:1, the
detection rate of Random Forest and Decision Tree are
8.69\% and 11.83\%, respectively. However, when we
applied {\sc PUDroid}, the detection rates increases to
71.51\% and 57.15\%, respectively. 
\vspace{2pt}    
\item {\sc PUDroid} with Random Forest yields higher
detection rates than those produced by {\sc PUDroid}
with Decision Tree. However, if the PU learning process is not
used, Decision Tree performs better than Random Forest.
For example, Decision Tree yields the detection rate
of 11.83\% while Random Forest only yields 8.69\% when
the ratio is 3:1.
\vspace{2pt}  
\item In extreme situations, Random Forest yields a
detection rate of 52.14\% when the ratio is 8:1. To put
this into perspective, this result is better than using
Random Forest alone when the ratio is 1:1.
\end{itemize}

In summary, {\sc PUDroid} is very effective in
identifying contaminants in highly contaminated
datasets.

\subsection{RQ 3: Unknown Contaminants} 

In the dataset, we have 178 families of malware, but
many families only have one malware sample. 
We only found that only 5\% of the total families contains
more than 100 malware and about 50\% of families
contain 3 samples or fewer.  Such a distribution
pattern indicates that some families of malware are
less likely to be repackaged and redeployed as new
variations while others are popularly repackaged and
redistributed as new variations. 

While the focus of this particular investigation is on
\textsc{PUDroid}'s performance when the benign dataset is 
contaminated with unknown malware. we also want the contaminants
to be representative of the current practice of
creating variations of malicious behaviors to promote
rapid infections of devices. As such, our
contaminants are from families that have more than 300
samples in our malware dataset. In this case, the
families include \emph{FakeInstaller},
\emph{DroidKungFu}, \emph{Plankton}, \emph{Opfake},
\emph{GinMaster} and \emph{BaseBridge}. 
Next, we briefly describe the malicious behaviors of
these families. 

\begin{figure*}[tbh!]
\begin{center}
\mbox{
\begin{minipage}[t]{.33\textwidth}
\includegraphics[width=1.00\textwidth]{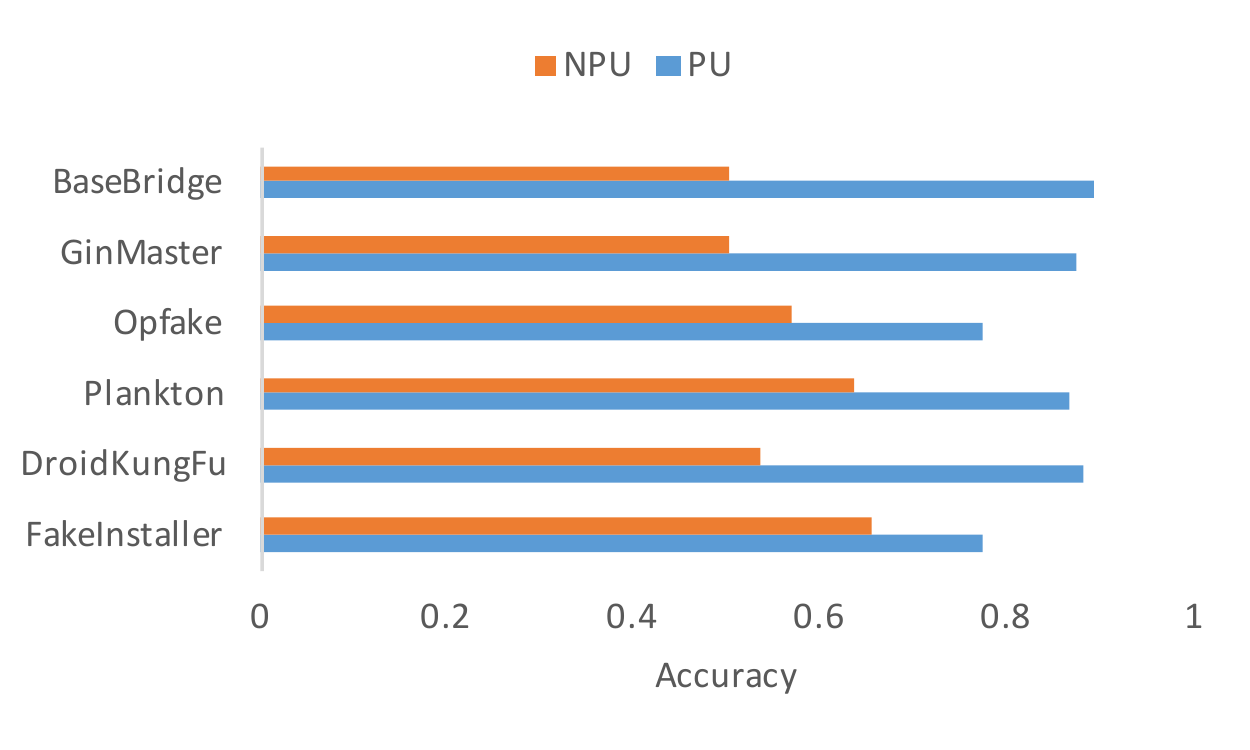}
\end{minipage}
\begin{minipage}[t]{.33\textwidth}
\includegraphics[width=1.0\textwidth]{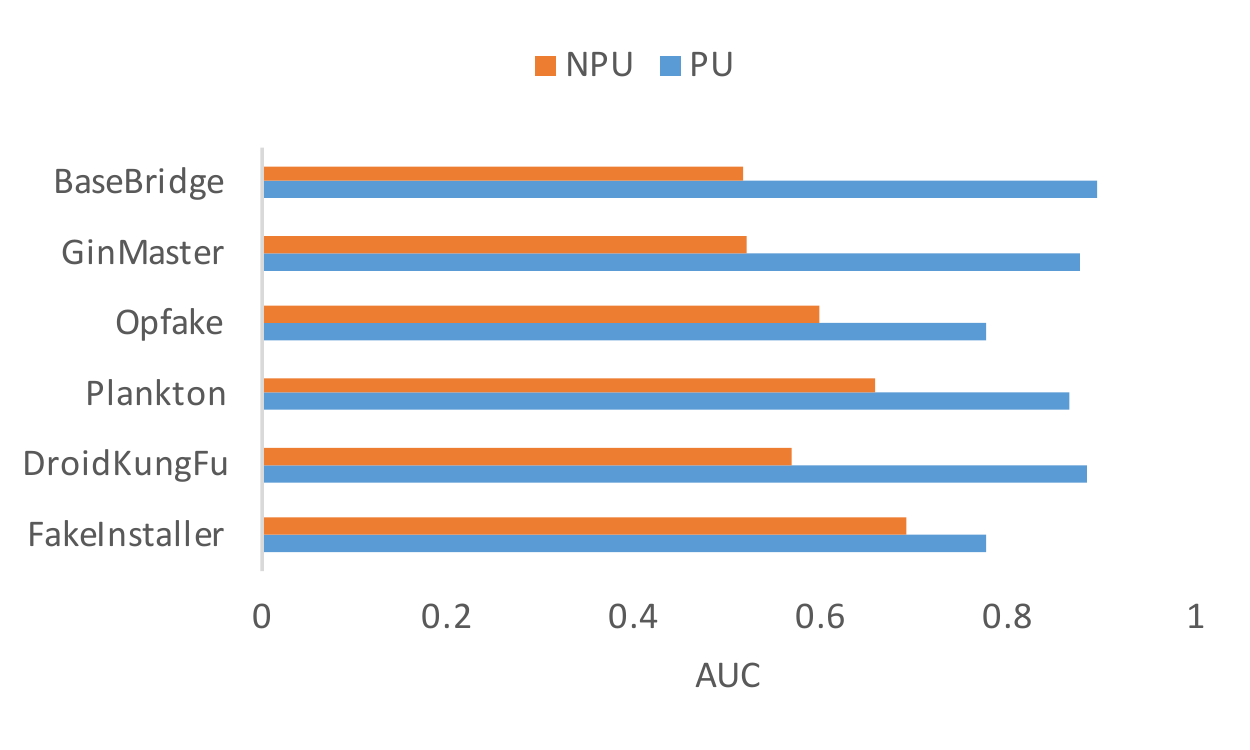}
\end{minipage}
\begin{minipage}[t]{.33\textwidth}
\includegraphics[width=1.0\textwidth]{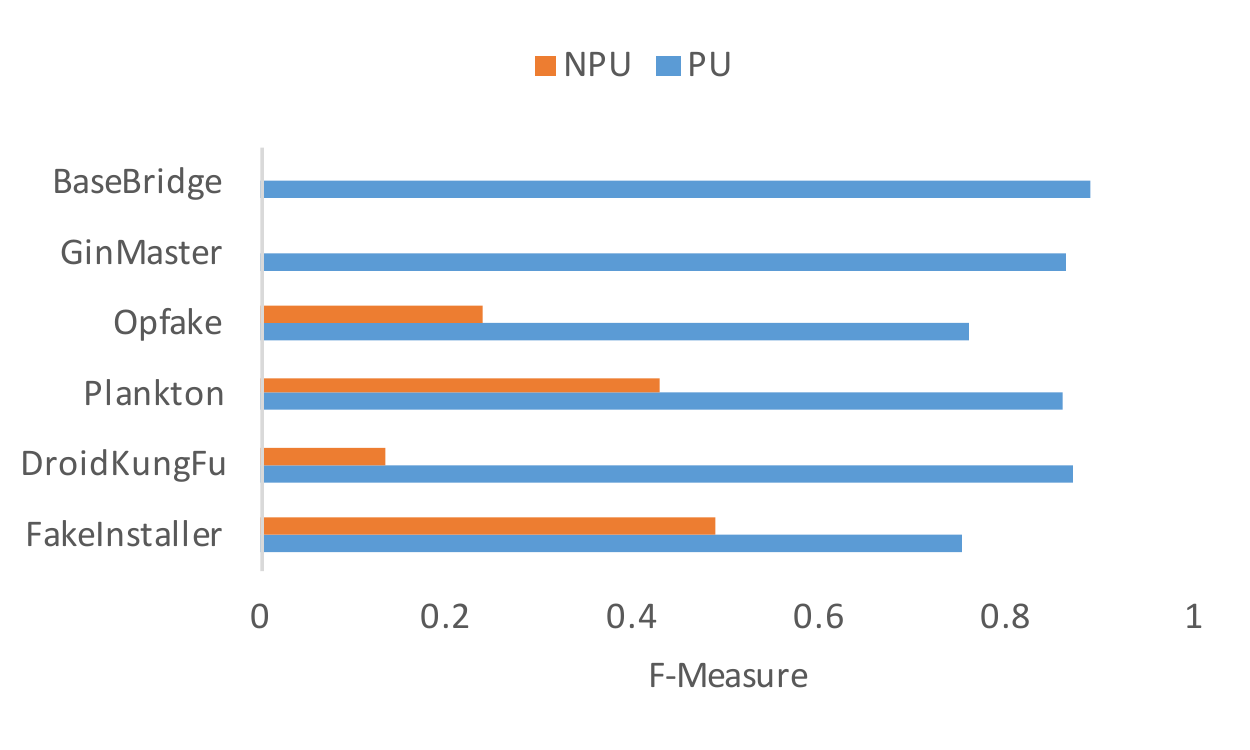}
\end{minipage}
}
\caption{Performance of PUDroid when faced with new malware contaminants}
\label{fig:results2}
\end{center}
\end{figure*}

\begin{figure*}[tbh!]
\begin{center}
\mbox{
\begin{minipage}[t]{.33\textwidth}
\includegraphics[width=1.00\textwidth]{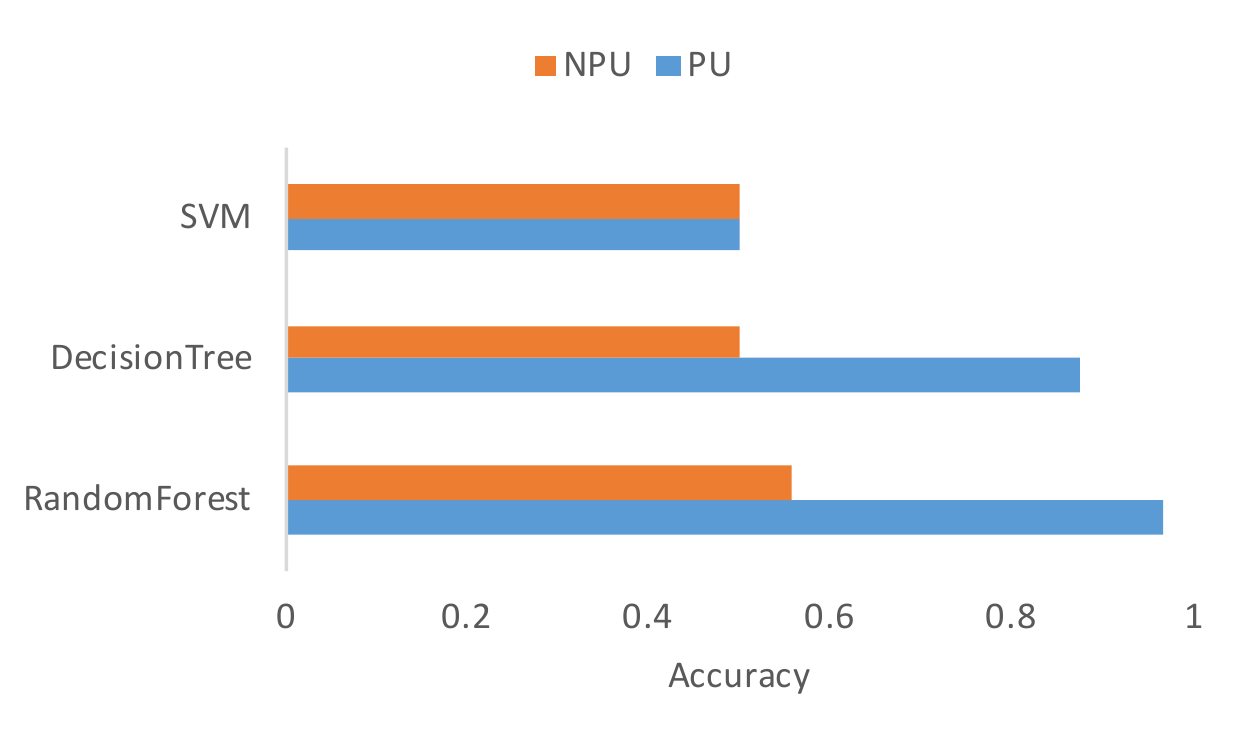}
\end{minipage}
\begin{minipage}[t]{.33\textwidth}
\includegraphics[width=1.0\textwidth]{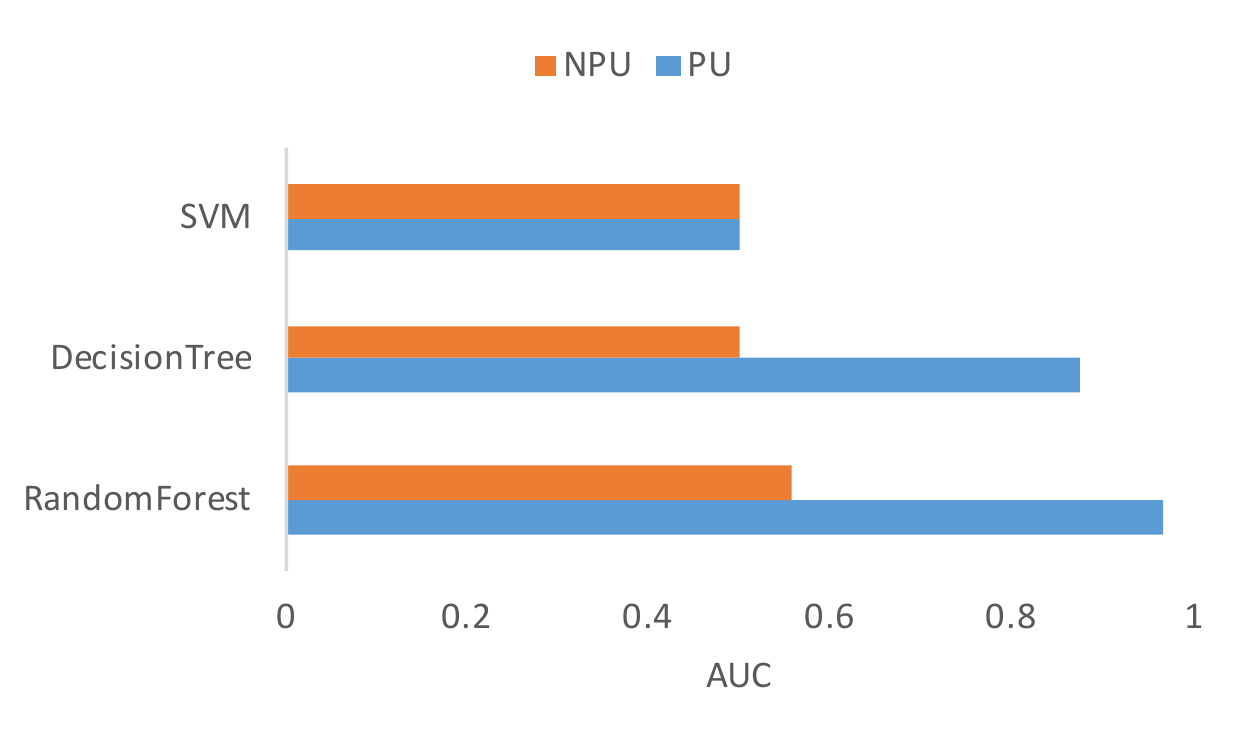}
\end{minipage}
\begin{minipage}[t]{.33\textwidth}
\includegraphics[width=1.0\textwidth]{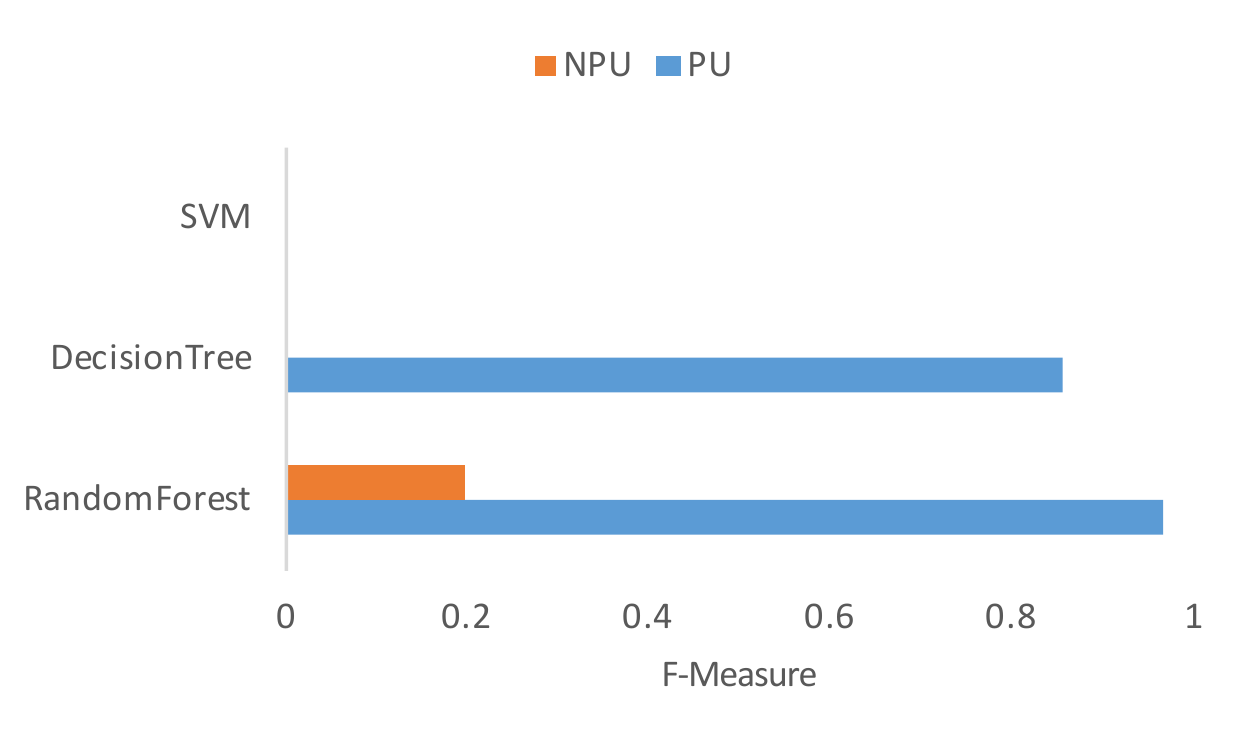}
\end{minipage}
}
\caption{Performance of PUDroid when faced with
contaminated data in malware dataset}
\label{fig:results3}
\end{center}
\end{figure*}


\begin{itemize}
\vspace{2pt}  
\item \emph{FakeInstaller} is the malware family with the
largest number of variations in our dataset (925
malware samples).  Malware authors simply repackaged
commonly distributed apps (e.g., \emph{Facebook}) with
malicious functionalities. These malicious apps send
SMS messages to premium rate numbers, without the
user’s consent, passing itself off as the installer for
a legitimate application. As previously reported, over
60\% of Android malware samples processed by McAfee
belong to this family~\cite{Ruiz2012}.
\vspace{2pt}  
\item \emph{DroidKungFu} family exploits several known
but unpatched vulnerabilities in earlier Android
versions to gain root access to a device and steal
sensitive information. Stolen information is sent to
remote command and control (C\&C) servers.
\vspace{2pt}  
\item \emph{Plankton} family focuses on to collect
every information on the smartphone, such as
International Mobile Equipment Identity (IMEI) number,
user ID, and the browser history. It then posts
information via URL.  It can also use remote control to
modify browser's bookmarks and install downloaded files
on devices.
\vspace{2pt}  
\item \emph{BaseBridge} family sends fake update
notifications to users in hope of fooling them to
install malicious software components to allow
cyber-criminals to remotely control infected devices.

\end{itemize}

We included five out of six malware families in the
malicious dataset for training. To create a
contaminated dataset, we mixed samples from the
remaining malware family and benign apps using the
ratio of 1:1.  Again, 1/3 of our dataset was used for
testing and 2/3 was used for training.  We repeated this
process so that each of the six malware families was
used as contaminants.  We also adopted SVM classifier in this
experiment because our earlier study indicates that SVM
can work very well when the number of contaminants is
moderate. We report the results in
Fig.~\ref{fig:results2}. 


As shown, we can clearly see that {\sc PUDroid} shows much better performance
than simply applying SVM without PU learning when the
system face any new malware contaminants especially in
BaseBridge, GinMaster and DroidKungFu. In BaseBridge
and GinMaster families, {\sc PUDroid} achieves more
than 45\% accuracy. Moreover, without {\sc PUDroid}
achieves 0\% in F-measure in families with smallest
samples among the six families (Basebridge and
Ginmaster). In the case of FakeInstallers, SVM alone is
effective since this is a very popular malware with the
most common behavior. With {\sc PUDroid}, we can
achieve 10\% higher accuracy. In general, {\sc PUDroid}
is effective in detecting repackaged contaminants.

\subsection{RQ 4: Detecting Benign Contaminants
in Malicious Datasets}

It is quite possible that benign apps could be also
mistakenly labeled as malicious and thus, are included
in the malicious datasets.  This is because to keep up
with new malware samples, the datasets need to be
updated periodically and this process can result in
mislabeling by researchers (e.g., the number of benign
apps is typically much larger than that of malicious
apps).  Furthermore, it is also possible for a malware
repository to be intentionally compromised by someone
(e.g., a cyber-attack to weaken training data).
Regardless of the causes, the presence of benign apps
in malicious datasets can make machine learning less
effective.

To evaluate this scenario, we again used 1/3 of our
dataset as the testing dataset. We set the ratio of
benign and malicious app as 8:1 to simulate the case of
highly contaminated datasets. In
Fig.~\ref{fig:results3}, we report the performances of
{\sc PUDroid} and using classifier alone. {\sc PUDroid}
with Random Forest achieves 96.36\% accuracy when a
malicious dataset is heavily contaminated with benign
apps. {\sc PUDroid} with Decision Tree is second.
However, the results show that Random Forest without PU
learning can still achieve 55.5\% accuracy, but
Decision Tree without PU learning achieves only about
50\%, which is the same as making a random guess. The
results also confirm prior observation that SVM does
not work well when the dataset is highly contaminated. 


%% file: 07-relatedwork.tex
\section{Related Work}\label{relatedwork}
In this work, we combine several concepts and
techniques to construct {\sc PUDroid}, a framework to
detect the malicious contaminants to increase accuracy
of any machine-learning based detector. The framework
performs feature selection and feature embedding
techniques to increase robustness and efficiency of our
proposed system.

Machine learning techniques have been used to build
several robust and effective malware detection
systems~\cite{arp2014drebin,sun2016sigpid}.
For example, {\sc SIGPID} \cite{sun2016sigpid} applies
67 machine learning algorithms to find which algorithm
is better at classification based on permission
features to detect the malware. Huang el
al. \cite{huang2013performance} explore the use of
machine learning of permission to detect malicious
applications. Their detector is based on four common
machine learning techniques. As these prior efforts
have shown, different machine learning algorithms
perform differently on different dataset, however, in
general, machine learning has been effective in
performing malware detection especially when datasets
are properly labeled.



Typically, more features can help to improve the
accuracy and detection rate, but also incur more time
and space. Many advanced techniques, i.e. tensor,
factorized machine \cite{he2014dusk, he2017kernelized, he2017multi},
improve the performance by leveraging multi-view datasets.
These techniques are frequently applied in many areas,
such as nature language \cite{zhang2016multi}, recommendation \cite{zheng2017joint}, bio-medical \cite{cao2017t}, images \cite{hao2013linear}, influence networks \cite{shao2015clustering}, behavioral detection \cite{sun2017sequential}.

Others have also used feature
selection in different
areas~\cite{sun2016sigpid,wei2017unsupervised,cao2014tensor,wei2017multi}.
By finding important or significant patterns, feature
selection can improve the overall performance of a
system. For example, {\sc SIGPID} \cite{sun2016sigpid}
proposes a multi-level mining technique to do feature
selection, which ends up using only 22 of 135
permissions to detect malware. 


In this work, we apply PU learning to help detect
contaminants.  PU learning is one kind of semi-supervised learning
have been used to perform link prediction in the social
network~\cite{ZYZ14}, and images \cite{cui2017inverse}. 
{\sc MLI}~\cite{ZYZ14} uses the
PU learning to help them deal with unlabeled link
prediction in social networks. To use these unlabeled
link information, {\sc MLI} \cite{ZYZ14} use PU learning
to identify reliable negative instances from the
unlabeled set with the spy technique \cite{LDLLY03}.
Our work, on the other hand, proposes to model the
\textit{hidden malware probability} based on the
identified positive and unlabeled sets. We propose a
non-parametric learning framework to detect malicious
contaminants.


%% file: 08-conclusion.tex
\section{Conclusion}\label{conlusion}

In this paper, we introduce {\sc PUDroid}, a framework
to detect and remove contaminants from training
datasets used in machine learning based malware
detection systems. We experimented with using a corpus
of over 5,500 malware and mixed them with benign apps
to create various sizes of unlabeled datasets.  We then
evaluate the performance of \textsc{PUDroid} under four
realistic settings. First, we evaluate the
effectiveness of \textsc{PUDroid} using three commonly
used classification techniques (SVM, Decision Trees,
and Random Forest). Second, we evaluate
\textsc{PUDroid} using  highly contaminated benign
datasets. Third, it is used to detect benign datasets
contaminated with unknown malware.  Fourth, we also
consider a case in which benign apps are intentionally
or accidentally included in the malware datasets.  

We then compare the detection performances of a system
with \textsc{PUDroid} and one without when contaminated
datasets are used for training.  The results indicate
that {\sc PUDroid} is effective at detecting malicious
contaminants.  {\sc PUDroid} can improve malware
detection rate by 62.82\% over detectors that use only
classifiers and no PU learning.  We also observe that
the detection accuracy is improved by 45\% when a
dataset is contaminated. To improve performance of {\sc
PUDroid}, we also apply feature selection and embedding
features. For future work, we plan to explore other
techniques to further improve the feature selection and
embedding process.  

\section*{Acknowledgements}
This work is supported in part by NSF through grants IIS-1526499 and CNS-1626432, NSFC through grants 61503253, 61672357 and 61672313, NIH through grant R01-MH080636, and the Science Foundation of Shenzhen through grant JCYJ20160422144110140.